\documentclass[12pt,reqno]{article}%
\usepackage{graphicx}
\usepackage{amsmath}
\usepackage{amsfonts}
\usepackage{amssymb}
\usepackage{algorithm}
\usepackage{algorithmic}%
\providecommand{\U}[1]{\protect\rule{.1in}{.1in}}
\newtheorem{theorem}{Theorem}[section]

\newtheorem{lemma}[theorem]{Lemma}

\newenvironment{proof}[1][Proof]{\noindent\textbf{#1.} }{\ \rule{0.5em}{0.5em}}

\begin{document}

\title{Bad Communities with High Modularity}
\author{Ath. Kehagias and L. Pitsoulis\\{\small Faculty of Engineering}\\{\small Aristotle Univ. of Thessaloniki}\\\texttt{kehagiat@auth.gr}, \texttt{pitsoulis@auth.gr}}
\date{\today}
\maketitle

\begin{abstract}
In this paper we discuss some problematic aspects of Newman's modularity
function $Q_{N}$. Given a graph $G$, the modularity of $G$ can be written as
$Q_{N}=Q_{f}-Q_{0}$, where $Q_{f}$ is the \emph{intracluster edge fraction} of
$G$ and $Q_{0}$ is the expected intracluster edge fraction of the \emph{null
model}, i.e., a randomly connected graph with same expected degree
distribution as $G$. It follows that the maximization of $Q_{N}$ must
accomodate two factors pulling in opposite directions:$\ Q_{f}$ favors a small
number of clusters and $Q_{0}$ favors many balanced (i.e., with approximately
equal degrees) clusters. In certain cases the $Q_{0}$ term can cause
\emph{over}estimation of the true cluster number; this is the opposite of the
well-known \emph{under}estimation effect caused by the \textquotedblleft%
\emph{resolution limit}\textquotedblright\ of modularity. We illustrate the
overestimation effect by constructing families of graphs with a
\textquotedblleft natural\textquotedblright\ community structure which,
however, does not maximize modularity. In fact, we prove that we can always
find a graph $G$ with a \textquotedblleft natural clustering\textquotedblright%
\ $\mathbf{V}\ $ of $G$ and another, balanced clustering $\mathbf{U}$ of $G$
such that (i)\ the pair $\left(  G,\mathbf{U}\right)  $ has higher modularity
than $\left(  G,\mathbf{V}\right)  $ and (ii)$\ \mathbf{V}$ and $\mathbf{U}$
are arbitrarily different.

\end{abstract}

%Arbitrarily Bad Communities \\Obtained by Modularity Maximization}

\section{Introduction\label{sec01}}

This paper describes some problems which may arise in using \emph{Newman's
modularity function} $Q_{N}$ for \emph{community detection}. Modularity is one
of the most popular quality functions in the community detection literature.
It is not only used to evaluate the community structure of a graph, but also
to perform community detection by \emph{modularity maximization}. However, it
is well known that modularity maximization can, in certain cases, yield the
\textquotedblleft wrong\textquotedblright\ community decomposition. Previous
work on this aspect has focused on the \emph{modularity resolution limit},
which causes \emph{under}estimation of the true number of communities. In this
paper we focus on the opposite effect, in other words we show that, in certain
cases, modularity maximization can \emph{over}estimate the number of communities.

In Section \ref{sec02} we present our nomenclature and notation; let us stress
from the beginning that we will use \textquotedblleft\emph{cluster}%
\textquotedblright\ as a synonym of \textquotedblleft
community\textquotedblright\ and \textquotedblleft\emph{clustering}%
\textquotedblright\ to denote both a partition of the nodes of a graph and the
activity of creating such a partition.

In Section \ref{sec03} we present an interpretation of $Q_{N}$ which, as far
as we know, has not been discussed previously. It is well known that the
modularity of a graph $G$ can be written in the form $Q_{N}$ $=Q_{f}$ $-Q_{0}%
$, where $Q_{f}$ is the \emph{intracluster edge fraction} of $G$ and $Q_{0}$
is the expected intracluster edge fraction of the \emph{null model}, i.e., a
graph $G^{\prime}$ which has the same expected degree distribution as $G$ but
randomly distributed edges. As explained in Section \ref{sec0302},
maximization of $Q_{f}$ favors clusterings with a small number of clusters and
few edges across clusters. On the other hand, as explained in Section
\ref{sec0303}, minimization of $Q_{0}$ favors clusterings with a large number
of clusters and each cluster having approximately equal degree (we call these
\textquotedblleft balanced clusterings\textquotedblright). \emph{Cluster
number selection} is performed by balancing these two opposite effects in the
maximization of $Q_{N}$.\ 

In Section \ref{sec0401} we exploit the behavior of $Q_{0}$ and construct
examples in which modularity maximization yields arbitrarily inaccurate
clusterings. More specifically, we construct a class of graphs $G_{K,N_{1}%
,N_{2}K}$ (where $K$, $N_{1}$, $N_{2}$ are parameters of the graph)\ with the
following properties.

\begin{enumerate}
\item Each graph $G_{K,N_{1},N_{2}}$ has a \textquotedblleft
natural\textquotedblright\ clustering $\mathbf{V}_{K,N_{1},N_{2}}$ (which,
however, does not maximize modularity).

\item We can find graphs $G_{K,N_{1},N_{2}}$ and clusterings $\mathbf{U}%
_{K,N_{1},N_{2},J}$ such that, by appropriate selection of $K,N_{1},N_{2}$ and
$J$, the following hold:

\begin{enumerate}
\item The pair $\left(  G_{K,N_{1},N_{2}},\mathbf{U}_{K,N_{1},N_{2},J}\right)
$ has higher modularity than the pair $\left(  G_{K,N_{1},N_{2}}%
,\mathbf{V}_{K,,N_{1},N_{2}}\right)  $;

\item The modularity of $\left(  G_{K,N_{1},N_{2}},\mathbf{U}_{K,N_{1}%
,N_{2},J}\right)  $ can become (by appropriate selection of $J$) arbitrarily
close to one;

\item the \emph{Jaccard similarity }between clusterings $\mathbf{V}%
_{K,N_{1},N_{2}}$ and $\mathbf{U}_{K,N_{1},N_{2},J}$ can become (by
appropriate selection of $J$) arbitrarily close to zero (hence $\mathbf{V}%
_{K,N_{1},N_{2}}$ and $\mathbf{U}_{K,N_{1},N_{2},J}$ are arbitrarily different
in the Jaccard sense).
\end{enumerate}
\end{enumerate}

\noindent We prove similar results for another class of graphs in Section
\ref{sec0402}.

In Section \ref{sec05} we discuss the implications of our results and
(previously published)\ related work by other authors. It is often claimed
that \ community detection by modularity maximization should be preferred over
other community detection methods because it does not require knowing the
number of clusters in advance. However, in light of our results (as well as
the previously known modularity resolution limit) this claim appears
ill-founded. We conclude the paper with a discussion of alternative cluster
number selection methods.

\section{Preliminaries\label{sec02}}

\begin{enumerate}
\item A \emph{graph }$G$ is a pair $\left(  V,E\right)  $, where $V$ is the
\emph{node set} (we will always assume $V=\left\{  1,2,...,n\right\}  $; hence
the number of nodes is $n=\left\vert V\right\vert $)\ and $E\subseteq\left\{
\left\{  u,v\right\}  :u,v\in V\right\}  $ is the \emph{edge set} (and
$m=\left\vert E\right\vert $ is the number of edges). Hence we will deal with
finite graphs without multiple edges and loops.

\item The \emph{adjacency matrix} of $G$ is an $n\times n$ matrix $A$ with
$A_{u,v}=1$ iff $\left\{  u,v\right\}  \in E$ and 0 otherwise. There is a
one-to-one correspondence between a graph $G$ and its adjacency matrix $A$.

\item A \emph{clustering} of $G=\left(  V,E\right)  $ is a partition
$\mathbf{V}\mathcal{=}\left\{  V_{1},...,V_{K}\right\}  $ of $V$. The
\emph{clusters} are the node sets $V_{1},...,V_{K}$, which satisfy $\cup
_{k=1}^{K}V_{k}=V$ and $\forall k,l:V_{k}\cap V_{l}=\emptyset$. The
\emph{size} of the clustering is $K$, the number of clusters. Given a graph
$G=\left(  V,E\right)  $, we denote by $\mathcal{V}$ the set of all
clusterings of $V$ and by $\mathcal{V}_{K}$ the set of clusterings of size
$K$. Sometimes we call $V_{k}$ a \emph{community}; this is simply a synonym of
\textquotedblleft cluster\textquotedblright.

\item Given a clustering $\mathbf{V}\mathcal{=}\left\{  V_{1},...,V_{K}%
\right\}  $ of the graph $G=\left(  V,E\right)  $, we define the following
edge sets (for $k=1,...,K$):%
\[
E_{k}=\left\{  \left\{  u,v\right\}  :u,v\in V_{k}\text{ and }\left\{
u,v\right\}  \in E\right\}  ,
\]
i.e., $E_{k}$ is the set of edges with both ends being nodes of $V_{k}$. The
edges contained in $\cup_{k=1}^{K}E_{k}$ are the \emph{intracluster }edges;
the remaining edges, i.e., the ones contained in $E-$ $\cup_{k=1}^{K}E_{k}$
are the \emph{extracluster }edges.

\item The \emph{degree function} $\deg\left(  \cdot\right)  :V\rightarrow
\mathbb{Z}$ is defined as follows: for any $v\in V$, $\deg\left(  v\right)
=\left\vert \left\{  \left\{  v,w\right\}  :\left\{  v,w\right\}  \in
E\right\}  \right\vert $ is the number of edges incident on $v$; we also
define, for any $U\subseteq V$, $\deg\left(  U\right)  =\sum_{v\in U}%
\deg\left(  v\right)  $, i.e., the sum of degrees of the nodes contained in
$U$.

\item The \emph{Jaccard similarity index }is defined as follows. Given any two
clusterings $\mathbf{W}_{1}$\textbf{, }$\mathbf{W}_{2}$ define%
\begin{align*}
a_{11}  &  =\text{\textquotedblleft num. of node pairs }\left\{  u,v\right\}
\text{ in same cluster under }\mathbf{W}_{1}\text{ and same cluster under
}\mathbf{W}_{2}\text{\textquotedblright;}\\
a_{10}  &  =\text{\textquotedblleft num. of node pairs }\left\{  u,v\right\}
\text{ in same cluster under }\mathbf{W}_{1}\text{ and different cluster under
}\mathbf{W}_{2}\text{\textquotedblright;}\\
a_{01}  &  =\text{\textquotedblleft num. of node pairs }\left\{  u,v\right\}
\text{ in different cluster under }\mathbf{W}_{1}\text{ and same cluster under
}\mathbf{W}_{2}\text{\textquotedblright.}%
\end{align*}
Then the Jaccard similarity index $S\left(  \mathbf{W}_{1},\mathbf{W}%
_{2}\right)  $ is defined by%
\[
S\left(  \mathbf{W}_{1},\mathbf{W}_{2}\right)  =\frac{a_{11}}{a_{10}%
+a_{01}+a_{11}}.
\]
$S\left(  \mathbf{W}_{1},\mathbf{W}_{2}\right)  $ takes values in $\left[
0,1\right]  $; values close to 1 show that $\mathbf{W}_{1},\mathbf{W}_{2}$ are
very similar; values close to 0 that they are very different.
\end{enumerate}

\section{An Intepretation of Modularity\label{sec03}}

\subsection{Modularity\label{sec0301}}

Given a graph $G=\left(  V,E\right)  $ with adjacency matrix $A$, we denote
the modularity of a clustering $\mathbf{V}$ by $Q_{N}\left(  \mathbf{V}%
,G\right)  $ and, following \cite{newman2004finding}, we define it by
\begin{equation}
Q_{N}\left(  \mathbf{V},G\right)  =\frac{1}{2m}\sum_{i,j\in V}\left(
A_{ij}-\frac{\deg\left(  i\right)  \deg\left(  j\right)  }{2m}\right)
\Delta\left(  i,j\right)  , \label{eq03001}%
\end{equation}
where $\Delta\left(  i,j\right)  $ equals one if $i$ and $j$ belong to the
same cluster and zero otherwise. Our notation emphasizes that $Q_{N}\left(
\mathbf{V},G\right)  $ is a function of both the graph and the clustering.

The motivation for introducing modularity can be seen by the following
interpretation: $Q_{N}\left(  \mathbf{V},G\right)  $\emph{ measures the
fraction of intracluster edges in }$G$\emph{ minus the expected value of the
same quantity in a graph }$G^{\prime}$\emph{ with the same clusters but random
connections between the nodes}\footnote{This is a paraphrase of Newman and
Girvan's description of modularity \cite[Section IV]{newman2004finding}%
.\label{ftnt01}}. $G^{\prime}$ is often called the \emph{null model}. Note
that the intracluster edge fraction of both $G$ and $G^{\prime}$ is computed
with respect to $\mathbf{V}$. A large value of $Q_{N}\left(  \mathbf{V}%
,G\right)  $ indicates that, with respect to $\mathbf{V}$, $G$ is quite
different from the null model; this is taken as evidence of $G$ having
\textquotedblleft\emph{strong community structure}\textquotedblright\ which is
\textquotedblleft well captured\textquotedblright\ by $\mathbf{V}$. Hence
modularity is a \emph{clustering quality function }(CQF) in the sense of
\cite{fortunato2010community}.

Other interpretations of modularity are possible; we will propose one a little
later. But first let us note that, in addition to characterizing a single
$\left(  \mathbf{V},G\right)  $ pair, modularity can be used to \emph{compare}
clusterings: \emph{by definition}, $\mathbf{V}$ is a better clustering of $G$
than $\mathbf{V}^{\prime}$ iff $Q_{N}\left(  \mathbf{V},G\right)
>Q_{N}\left(  \mathbf{V}^{\prime},G\right)  $\textbf{. }Taking this one step
further, $\mathbf{V}^{\ast}=\arg\max_{\mathbf{V}}Q_{N}\left(  \mathbf{V}%
,G\right)  $ is the \emph{best} clustering of $G$. This has two implications:
first, a large value of $\max_{\mathbf{V}}Q_{N}\left(  \mathbf{V},G\right)  $
indicates that $G$ has strong community structure and second, modularity
maximization can be used to obtain graph clusterings, i.e., perform community
detection; this has been the basis of a large number of community detection algorithms.

While modularity maximization is a very popular method for community
detection, it also has shortcomings which have been widely reported in the
literature. For example, the \emph{modularity resolution limit} has attracted
a lot of attention \cite{fortunato2007resolution,good2010performance}; we will
discuss it in Section \ref{sec0303}. But first let us note what appears to be
a more basic limitation of modularity. As already mentioned, a large
$Q_{N}\left(  \mathbf{V},G\right)  $ value indicates strong community
structure and good clustering; but what \emph{is} a \textquotedblleft large
$Q_{N}\left(  \mathbf{V},G\right)  $ value\textquotedblright? While it is
known \cite{brandes2007finding} that $-\frac{1}{2}\leq Q_{N}\left(
\mathbf{V},G\right)  \leq1$ for every pair $\left(  \mathbf{V},G\right)
$,$\ $ examples appear in the community detection literature
\cite{fortunato2010community} which have strong (intuitively
perceived)\ community structure and yet their maximum modularity is closer to
zero than to one. Graphs of high modularity and weak community structure have
also been reported \cite{bagrow2012communities,fortunato2010community}.

A frequently proposed explanation for the shortcomings of modularity is that
the use of the null model is not well justified \cite{fortunato2010community}.
In Section \ref{sec0303} we will consider an alternative, complementary
explanation. But first we will examine another CQF.

\subsection{Intracluster Edge Fraction\label{sec0302}}

A popular characterization of a graph community is that \textquotedblleft%
\emph{there must be more edges `inside' the community than edges linking
vertices of the community with the rest of the graph}\textquotedblright%
\ \cite[Section III-B.1]{fortunato2010community}. Variations of this principle
have been stated by several authors\footnote{An extreme statement of this idea
appears in \cite{chen2007checking}: \textquotedblleft a community network
$G_{0}=(V,E_{0})\ $[is] a graph $G_{0}$ that is a disjoint union of complete
subgraphs\textquotedblright.}.

A \emph{prima facie} reasonable way to quantify the principle is through the
\emph{intracluster edge fraction}, denoted by $Q_{f}\left(  \mathbf{V}%
,G\right)  $ and defined by%
\begin{equation}
Q_{f}\left(  \mathbf{V},G\right)  =\frac{\sum_{k=1}^{K}\left\vert
E_{k}\right\vert }{m}. \label{eq03003}%
\end{equation}
For every $G$ and $\mathbf{V}$, $Q_{f}\left(  \mathbf{V},G\right)  \in\left[
0,1\right]  $. A high (i.e., close to 1)\ value of $Q_{f}\left(
\mathbf{V},G\right)  $ indicates that the pair $\left(  \mathbf{V},G\right)  $
has many intracluster and few extracluster edges.

Unfortunately, a high $Q_{f}\left(  \mathbf{V},G\right)  $ value does
\emph{not} guarantee either that $G$ has strong community structure or that
$\mathbf{V}$ is a good clustering of $G$. Indeed we can always achieve the
maximum value $Q_{f}\left(  \mathbf{V},G\right)  =1$ by taking $\mathbf{V=}%
\left\{  V\right\}  $ (i.e., the unique clustering of size one)\ but this
tells us nothing about the \textquotedblleft true\textquotedblright\ community
structure of $G$. This observation can be generalized. First define the
following function%
\begin{equation}
F_{G}\left(  K\right)  =\max_{\mathbf{V\in}\mathcal{V}_{K}}Q_{f}\left(
\mathbf{V},G\right)  . \label{eq03004}%
\end{equation}
In words, for a given graph $G$, $F_{G}\left(  K\right)  $ is the maximum
intracluster edge fraction \emph{achieved by clusterings of size }$K$. Now we
can prove the following.

\begin{theorem}
\label{prp0301}For any graph $G=\left(  V,E\right)  $, $F_{G}\left(  K\right)
$ is a nonincreasing function of $K$.
\end{theorem}

\begin{proof}
There exists a single clustering of size one, namely\ $\mathbf{V}^{\left(
1\right)  }\mathbf{=}\left\{  V\right\}  $. Denote the set of intracluster
edges by $E_{1}^{\left(  1\right)  }$; obviously $E_{1}^{\left(  1\right)
}=E$ (i.e., \emph{all} edges are intracluster). Hence $F_{G}\left(  1\right)
=\frac{\left\vert E_{1}^{\left(  1\right)  }\right\vert }{\left\vert
E\right\vert }=1$.

Let $\mathbf{V}^{\left(  K\right)  }\mathbf{=}\left\{  V_{1}^{\left(
K\right)  },V_{2}^{\left(  K\right)  },...,V_{K}^{\left(  K\right)  }\right\}
$ be the optimal clustering of size $K$; the intracluster edge sets are
$E_{1}^{\left(  K\right)  }$, ..., $E_{K}^{\left(  K\right)  }$. Create a
clustering $\mathbf{V}^{\prime}$ of size $K-1$ by merging $V_{K-1}^{\left(
K\right)  }$ and $V_{K}^{\left(  K\right)  }$. In other words
\[
\mathbf{V}^{\prime}=\left\{  V_{1}^{\left(  K\right)  },V_{2}^{\left(
K\right)  },...,V_{K-2}^{\left(  K\right)  },V_{K-1}^{\left(  K\right)  }\cup
V_{K}^{\left(  K\right)  }\right\}  .
\]
Under $\mathbf{V}^{\prime}$ the intracluster edges are
\[
E_{1}^{\prime}=E_{1}^{\left(  K\right)  },\quad...,\quad E_{K-2}^{\prime
}=E_{K-2}^{\left(  K\right)  },\quad E_{k-1}^{\prime}.
\]
We have $E_{K-1}^{\left(  K\right)  }\cup E_{K}^{\left(  K\right)  }\subseteq
E_{K-1}^{\prime}$\ and $\left\vert E_{K-1}^{\left(  K\right)  }\right\vert
+\left\vert E_{K}^{\left(  K\right)  }\right\vert \leq\left\vert
E_{K-1}^{\prime}\right\vert $. Hence
\[
F_{G}\left(  K\right)  =Q_{f}\left(  \mathbf{V}^{\left(  K\right)  },G\right)
=\frac{\sum_{k=1}^{K}\left\vert E_{k}^{\left(  K\right)  }\right\vert
}{\left\vert E\right\vert }\leq\frac{\sum_{k=1}^{K-2}\left\vert E_{k}^{\left(
K\right)  }\right\vert }{\left\vert E\right\vert }+\frac{\left\vert
E_{K-1}^{\prime}\right\vert }{\left\vert E\right\vert }=Q_{f}\left(
\mathbf{V}^{\prime},G\right)  .
\]
But
\[
Q_{f}\left(  \mathbf{V}^{\prime},G\right)  \leq\max_{\mathbf{V\in}%
\mathcal{V}_{K}}Q_{f}\left(  \mathbf{V},G\right)  =F_{G}\left(  K-1\right)  .
\]
It follows that $0\leq F_{G}\left(  n\right)  \leq...\leq F_{G}\left(
2\right)  \leq F_{G}\left(  1\right)  =1\ $and the proof is complete.
\end{proof}

Hence, for any $G$, $Q_{f}\left(  \mathbf{V},G\right)  $ is maximized at $K=1$
and this gives us no information about the actual community structure of $G$.
In other words, Theorem \ref{prp0301} implies that $Q_{f}$ maximization cannot
determine the optimal number of clusters. On the other hand, if $K$ is given
in advance (as a parameter) then $\mathbf{V}^{\left(  K\right)  }=\arg
\max_{\mathbf{V\in}\mathcal{V}_{K}}Q_{f}\left(  \mathbf{V},G\right)  $ is a
reasonable candidate \emph{for the best clustering of size }$K$. This has
sometimes been phrased as a criticism of community detection by $Q_{f}$
maximization. For instance, in \cite{fortunato2010community} is stated that
\textquotedblleft\emph{Algorithms for graph partitioning are not good for
community detection, because it is necessary to provide as input the number of
groups}\textquotedblright. However, this criticism is valid only to the extent
that other algorithms exist which \emph{can} obtain the true number of groups
(clusters). For example, an alleged advantage of modularity is that its
maximization yields the correct number of clusters; let us now discuss this claim.

\subsection{Modularity as Augmented Intracluster Edge Fraction\label{sec0303}}

The claim that modularity maximization can determine the true number of
clusters has been put in doubt by the discovery of the modularity
\emph{resolution limit}. As explained in
\cite{fortunato2007resolution,good2010performance} and several other papers,
there exist graphs $G$ for which the clustering obtained by maximizing
modularity has fewer clusters than the \textquotedblleft intuitively
correct\textquotedblright\ clustering of $G$. In other words, modularity
maximization can \emph{under}estimate the number of clusters. We will now
argue that modularity maximization can also \emph{over}estimate the number of
clusters. Our argument is intuitive, but it will form the basis of some
precise results presented in Section \ref{sec04}.

Modularity can be computed by the formula (which is known to be equivalent to
(\ref{eq03001}) ):
\begin{equation}
Q_{N}\left(  \mathbf{V},G\right)  =\sum_{k=1}^{K}\frac{\left\vert
E_{k}\right\vert }{m}-\sum_{k=1}^{K}\left(  \frac{\deg\left(  V_{k}\right)
}{2m}\right)  ^{2}. \label{eq03002}%
\end{equation}
Defining
\begin{equation}
Q_{0}\left(  \mathbf{V},G\right)  =\sum_{k=1}^{K}\left(  \frac{\deg\left(
V_{k}\right)  }{2m}\right)  ^{2} \label{eq03006}%
\end{equation}
we can rewrite (\ref{eq03002})\ as
\begin{equation}
Q_{N}\left(  \mathbf{V},G\right)  =Q_{f}\left(  \mathbf{V},G\right)
-Q_{0}\left(  \mathbf{V},G\right)  . \label{eq03005}%
\end{equation}
Hence Newman's modularity is the difference of $Q_{f}\left(  \mathbf{V}%
,G\right)  $ and the auxiliary function $Q_{0}\left(  \mathbf{V},G\right)  $.
As already mentioned, the introduction of $Q_{0}\left(  \mathbf{V},G\right)  $
is usually motivated by appeal to the null model \cite{newman2004finding}; we
will now present an alternative, complementary view.

Suppose momentarily that $K$ is given and we want to \emph{minimize}
$Q_{0}\left(  \mathbf{V},G\right)  $ with respect to $\mathbf{V}=\left\{
V_{1},...,V_{K}\right\}  $. For simplicity of notation, define $p_{k}%
=\frac{\deg\left(  V_{k}\right)  }{2m}$; then
\[
Q_{0}\left(  \mathbf{V},G\right)  =\sum_{k=1}^{K}\left(  \frac{\deg\left(
V_{k}\right)  }{2m}\right)  ^{2}=\sum_{k=1}^{K}p_{k}^{2}%
\]
and we also have
\[
\sum_{k=1}^{K}p_{k}=\sum_{k=1}^{K}\frac{\deg\left(  V_{k}\right)  }{2m}=1.
\]
Hence we want to solve the following problem:%
\begin{equation}
\text{given }K\text{, minimize }\sum_{k=1}^{K}p_{k}^{2}\text{\qquad\qquad
subject to}\text{: }0\leq p_{k}\leq1\text{ and }\sum_{k=1}^{K}p_{k}=1.
\label{eq03007}%
\end{equation}
Of course there are additional constraints on the $p_{k}$'s:\ each of them
must be obtained by summing the degrees of $V_{k}$, which is a set of nodes of
the given graph $G$. However, assume for the time being that the $p_{k}$'s are
continuously valued and must only satisfy the constraints of (\ref{eq03007})
(these assumptions will be removed a little later). Under these assumptions,
the solution to (\ref{eq03007}) is $p_{k}=\frac{1}{K}$ for all $k$; the
minimum thus achieved is $\frac{1}{K}$.

Next consider the problem:%
\begin{equation}
\text{minimize }\sum_{k=1}^{K}p_{k}^{2}\text{\qquad\qquad subject to}\text{:
}K\in\left\{  1,...,n\right\}  \text{, }0\leq p_{k}\leq1\text{ and }\sum
_{k=1}^{K}p_{k}=1. \label{eq03008}%
\end{equation}
We can solve (\ref{eq03008}) by first solving (\ref{eq03007}) separately for
each $K\in\left\{  1,...,n\right\}  $ and then looking for the overall
minimum; we see that this is $\frac{1}{n}$ and is achieved at $K=n$ and
$p_{k}=\frac{1}{n}$ for all $k$.

Going back to the minimization of $Q_{0}\left(  \mathbf{V},G\right)  $ we note
that, in general, the overall minimum $\sum_{k=1}^{K}p_{k}^{2}=\frac{1}{n}$
will only be achieved under very special circumstances. Namely, if all nodes
of $\mathbf{G}$ have equal degree, then
\[
\min_{\mathbf{V}\in\mathcal{V}}Q_{0}\left(  \mathbf{V},G\right)  =Q_{0}\left(
\mathbf{V}^{\ast},G\right)  =\frac{1}{n}%
\]
where $\mathbf{V}^{\ast}=\left\{  V_{1},...,V_{n}\right\}  $ and
$V_{i}=\left\{  i\right\}  $ for $i\in\left\{  1,...,n\right\}  $. But even
when the nodes of $G$ do not have equal degrees, it seems intuitively obvious
that small values of $Q_{0}\left(  \mathbf{V},G\right)  $ are achieved by
clusterings $\mathbf{V}$ which have many clusters (large $K$) and distribute
nodes between clusters so that $p_{k}=\frac{\deg\left(  V_{k}\right)  }{2m}$
is approximately the same for all $k\in\left\{  1,...,K\right\}  $. In
\ Section \ref{sec04} we will see precise examples which justify this intuition.

Let us now apply the above observations to modularity maximization. Since
(i)$\ Q_{N}=Q_{f}-Q_{0}$, (ii)$\ Q_{f}$ achieves its maximum at $K=1$ and
(iii)$\ $we expect $Q_{0}$ to achieve its minimum at or near $K=n$, we
conclude that the following factors will influence the outcome of modularity
maximization:\ the $Q_{f}$ term pulls $K$ towards small values and the $Q_{0}$
towards large ones; in addition the $Q_{f}$ term favors clusterings which
correspond to the \textquotedblleft natural\textquotedblright\ community
structure of $G$ (i.e., there exist few extracluster edges) while the $Q_{0}$
favors \textquotedblleft balanced\textquotedblright\ clusterings (i.e., each
cluster has more or less the same degree). The final outcome depends on (among
other factors) the relative magnitudes of $Q_{f}$ and $Q_{0}$.

These observations agree with previously published remarks, e.g., that
\textquotedblleft the existing modularity optimization method does not perform
well in the presence of unbalanced community structures\textquotedblright%
\ \cite{zhang2012community} and \textquotedblleft for modularity's null model
graphs, the modularity maximum corresponds to an equipartition of the
graph\textquotedblright\ \cite{fortunato2010community}.\ However, the above
works (and many other)\ concentrate on examples in which modularity
maximization underestimates the cluster number, while our analysis suggests an
overestimation effect. Since, to the best of our knowledge, overestimation
examples do not appear in the literature, we will present some in Section
\ref{sec04}.

Let us note, in concluding this section, that one method used to address the
modularity resolution limit is to introduce a modified modularity function.
This function is often written in the form%
\[
Q\left(  \mathbf{V},G;\gamma\right)  =Q_{f}\left(  \mathbf{V},G\right)
-\gamma Q_{0}\left(  \mathbf{V},G\right)
\]
\ where $\gamma$ is a \textquotedblleft tuning parameter\textquotedblright%
\ (see
\cite{alvarez2010weakly,krings2011upper,reichardt2006statistical,traag2011narrow,xiang2011limitation}
and also \cite{kumpula2007limited,li2008quantitative}). With $\gamma=1$,
$Q\left(  \mathbf{V},G;1\right)  =Q_{N}\left(  \mathbf{V},G\right)  $, the
original Newman's modularity. If this underestimates (resp.
overestimates)\ the \textquotedblleft true\textquotedblright\ number of
clusters, formation of more (resp. fewer)\ clusters can be encouraged by
increasing (resp. decreasing)$\ \gamma$ and hence the influence of the
$Q_{0}\left(  \mathbf{V},G\right)  $ term on the maximization problem.
However, it seems that no \textquotedblleft universally
correct\textquotedblright\ value of $\gamma$ exists; in other words, the
resolution limit can occur for any $\gamma$ value
\cite{traag2011narrow,xiang2011limitation}.

\section{Bad Clusterings with High Modularity\label{sec04}}

In this section we construct graphs admitting (i)\ a \textquotedblleft%
\emph{natural}\textquotedblright\ clustering and (ii)\ a sequence of
\textquotedblleft\emph{arbitrarily bad}\textquotedblright\ clusterings which
achieve higher modularity than the natural one. In fact, as we will see, the
arbitrarily bad clusterings can achieve modularity arbitrarily close to one
and they can be \textquotedblleft\emph{arbitrarily different}%
\textquotedblright\ from the natural clustering (we will presently explain
precisely what we mean by the terms \textquotedblleft
natural\textquotedblright,\ \textquotedblleft arbitrarily
bad\textquotedblright\ and \textquotedblleft arbitrarily
different\textquotedblright). These results indicate that, at least in certain
cases, modularity is not a good CQF.

\subsection{First Example\label{sec0401}}

To establish the abovementioned results, we will construct a family of graphs
$G_{K,N_{1},N_{2}}$ (where $K,N_{1},N_{2}$ are parameters)\ such that the
graph $G_{K,N_{1},N_{2}}$ has an easily recognized \textquotedblleft
natural\textquotedblright\ clustering $\mathbf{V}_{K,N_{1},N_{2}}$ (for every
$K,N_{1},N_{2}$).

We define $G_{K,N_{1},N_{2}}$ as follows. First, for any $N_{1}$, $N_{2}$ we
define the disconnected graph $G_{N_{1},N_{2}}$ to be the union of a path of
$N_{1}$ nodes and a path of $N_{2}$ nodes; second, we let the disconnected
graph $G_{K,N_{1},N_{2}}$ be the union of $K$ copies of $G_{N_{1},N_{2}}$. The
construction is illustrated in Figure \ref{fig0401}. \begin{figure}[h]
\centering\scalebox{0.6}{\includegraphics{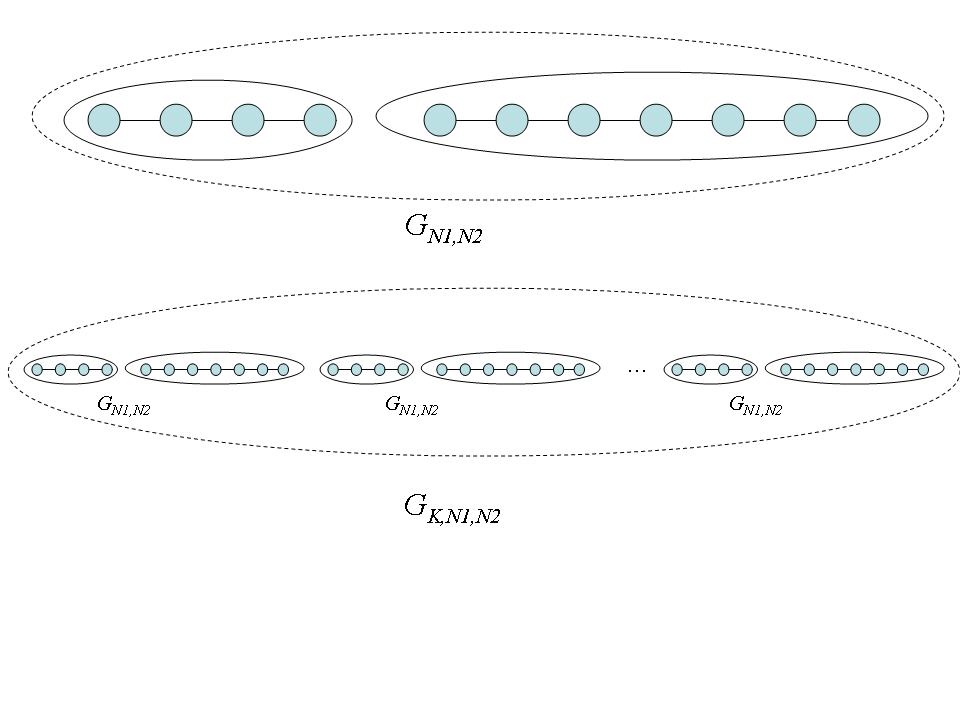}}\caption[Graph Family
A]{Graph Family $G_{K,N_{1},N_{2}}.$}%
\label{fig0401}%
\end{figure}

We claim that the \emph{natural }clustering of $G_{K,N_{1},N_{2}}$ is
$\mathbf{V}_{K,N_{1},N_{2}}$ = \{$V_{K,N_{1},N_{2},1}$, $V_{K,N_{1},N_{2},2}$,
$...$, $V_{K,N_{1},N_{2},2K}$\}, where $V_{K,N_{1},N_{2},k}$ is the node set
of the $k$-th connected component of $G$ (with $k\in\left\{
1,2,...,2K\right\}  $, see Figure \ref{fig0401}). At the risk of belaboring
the obvious, we note that, if $u\in V_{K,N_{1},N_{2},i}$ and $v\in
V_{K,N_{1},N_{2},j}$ and $i\neq j$, then there exists no path connecting $u$
and $v$; hence they should never be put in the same cluster. So the biggest
possible clusters are the $V_{K,N_{1},N_{2},i}$'s. On the other hand, there is
no justification for splitting some $V_{K,N_{1},N_{2},i}$ at any particular
edge, since all edges have the same connectivity pattern, i.e., the $i$-th
edge connects nodes $i$ and $i+1$. Hence $\mathbf{V}_{K,N_{1},N_{2}}$ is the
\textquotedblleft intuitively best\textquotedblright\ (i.e., the
\textquotedblleft natural\textquotedblright) clustering of $G_{K,N_{1},N_{2}}$.

\begin{lemma}
\label{prp0401}For every $K,N_{1},N_{2}\in\mathbb{N}$ with $N_{1},N_{2}%
\geq3\mathbb{\ }$and $J\leq n=K\left(  N_{1}+N_{2}\right)  $ we have%
\begin{equation}
Q_{N}\left(  \mathbf{V}_{K,N_{1},N_{2}},G_{K,N_{1},N_{2}}\right)
=1-\frac{\left(  N_{1}-1\right)  ^{2}+\left(  N_{2}-1\right)  ^{2}}{K\left(
N_{1}+N_{2}-2\right)  ^{2}}. \label{eq04002}%
\end{equation}

\end{lemma}

\begin{proof}
We fix $K,N_{1},N_{2}$ and, for brevity, we write $G$ for $G_{K,N_{1},N_{2}}$
and $\mathbf{V}$ for $\mathbf{V}_{K,N_{1},N_{2}}$. We have%
\[
Q_{N}\left(  \mathbf{V},G\right)  =\frac{\sum_{k=1}^{2K}\left\vert
E_{k}\right\vert }{m}-\frac{\sum_{k=1}^{2K}\left(  \deg\left(  V_{k}\right)
\right)  ^{2}}{\left(  2m\right)  ^{2}}.
\]
Under $\mathbf{V}$, $G$ has no extracluster edges hence we have%
\begin{equation}
\frac{\sum_{k=1}^{2K}\left\vert E_{k}\right\vert }{m}=1. \label{eq04003}%
\end{equation}
We can separate $\mathbf{V}$ into two subsets of clusters:\ $\mathbf{V}%
^{\prime}=\left\{  V_{1},V_{3},...,V_{2K-1}\right\}  $ contains the the
clusters with $N_{1}$ nodes and $\mathbf{V}^{\prime\prime}=\left\{
V_{2},V_{4},...,V_{2K}\right\}  $ contains the the clusters with $N_{2}$
nodes. Each $V_{k}\in\mathbf{V}^{\prime}$ has $N_{1}-2$ \textquotedblleft
inner nodes\textquotedblright\ of degree 2 and two \textquotedblleft border
nodes\textquotedblright\ of degree 1; similarly, each $V_{k}\in\mathbf{V}%
^{\prime\prime}$ has $N_{2}-2$ inner nodes and $2$ border nodes. Hence%
\begin{align*}
\forall &  :V_{k}\in\mathbf{V}^{\prime}:\deg\left(  V_{k}\right)  =2\left(
N_{1}-2\right)  +2=2\left(  N_{1}-1\right) \\
\forall &  :V_{k}\in\mathbf{V}^{\prime\prime}:\deg\left(  V_{k}\right)
=2\left(  N_{2}-2\right)  +2=2\left(  N_{2}-1\right)
\end{align*}
The total number of edges is
\[
m=\frac{\sum_{V_{k}\in\mathbf{V}}\deg\left(  V_{k}\right)  }{2}=\frac
{\sum_{V_{k}\in\mathbf{V}^{\prime}}\deg\left(  V_{k}\right)  +\sum_{V_{k}%
\in\mathbf{V}^{\prime\prime}}\deg\left(  V_{k}\right)  }{2}=K\left(
N_{1}+N_{2}-2\right)  .
\]
Also,
\begin{align}
\frac{\sum_{k=1}^{2K}\left(  \deg\left(  V_{k}\right)  \right)  ^{2}}{\left(
2m\right)  ^{2}}  &  =\frac{\sum_{V_{k}\in\mathbf{V}^{\prime}}\left(  2\left(
N_{1}-1\right)  \right)  ^{2}}{\left(  2K\left(  N_{1}+N_{2}-2\right)
\right)  ^{2}}+\frac{\sum_{V_{k}\in\mathbf{V}^{\prime\prime}}\left(  2\left(
N_{2}-1\right)  \right)  ^{2}}{\left(  2K\left(  N_{1}+N_{2}-2\right)
\right)  ^{2}}\nonumber\\
&  =\frac{K\left(  N_{1}-1\right)  ^{2}+K\cdot\left(  N_{2}-1\right)  ^{2}%
}{K^{2}\left(  N_{1}+N_{2}-2\right)  ^{2}}=\frac{\left(  N_{1}-1\right)
^{2}+\left(  N_{2}-1\right)  ^{2}}{K\left(  N_{1}+N_{2}-2\right)  ^{2}}.
\label{eq04004}%
\end{align}
Combining (\ref{eq04003}) and (\ref{eq04004})\ we get (\ref{eq04002}).
\end{proof}

Let us now introduce the \textquotedblleft bad clusterings\textquotedblright.
For every triple $\left(  K,N_{1},N_{2}\right)  $, we define a \emph{sequence}
$\left\{  \mathbf{U}_{K,N_{1},N_{2},J}\right\}  _{J=1}^{n}$ of clusterings of
$G_{K,N_{1},N_{2}}$. For a fixed $J$, let $L=\left\lfloor \frac{n}%
{J}\right\rfloor $; writing for brevity $\mathbf{U}_{J}$ in place of
$\mathbf{U}_{K,N_{1},N_{2},J}$, we let $\mathbf{U}_{J}=\left\{  U_{1}%
,...,U_{J},U_{J+1}\right\}  $ consist of the following $J+1$ clusters:%
\[
U_{1}=\left\{  1,...,L\right\}  \text{, }U_{2}=\left\{  L+1,...,2L\right\}
\text{, ... , }U_{J}=\left\{  \left(  J-1\right)  L+1,...,JL\right\}  \text{,
}U_{J+1}=\left\{  JL+1,...,n\right\}  ;
\]
if $n=JL$ then $U_{J+1}=\emptyset$. In other words, $\mathbf{U}_{J}$ contains
$J$ clusters each containing the same number of nodes (namely $L=\left\lfloor
\frac{n}{J}\right\rfloor $) and perhaps an additional cluster (with fewer than
$L$ nodes). Obviously $\mathbf{U}_{J}$ is a \textquotedblleft well
balanced\textquotedblright\ clustering.

\begin{lemma}
\label{prp0402}For every $K,N_{1},N_{2},J\in\mathbb{N}$ with $N_{1},N_{2}%
\geq3\mathbb{\ }$we have%
\begin{equation}
Q_{N}\left(  \mathbf{U}_{K,N_{1},N_{2},J},G_{K,N_{1},N_{2}}\right)
\geq1-\frac{1}{K\left(  N_{1}+N_{2}-2\right)  }J-\frac{2\left(  N_{1}%
+N_{2}\right)  ^{2}}{\left(  N_{1}+N_{2}-2\right)  ^{2}}J^{-1}.
\label{eq04011}%
\end{equation}

\end{lemma}

\begin{proof}
We write $G$ for $G_{K,N_{1},N_{2}}$ and $\mathbf{U}_{J}$ for $\mathbf{U}%
_{K,N_{1},N_{2},J}$. We have%
\[
Q_{N}\left(  \mathbf{U}_{J},G\right)  =\frac{\sum_{k=1}^{J+1}\left\vert
E_{k}\right\vert }{m}-\frac{\sum_{k=1}^{J+1}\left(  \deg\left(  U_{k}\right)
\right)  ^{2}}{\left(  2m\right)  ^{2}}.
\]
Consider first $\frac{\sum_{k=1}^{J+1}\left\vert E_{k}\right\vert }{m}$. A
little thought shows that $\mathbf{U}_{J}$ has at most $J+1$ clusters and $J$
extracluster edges. Hence%
\begin{equation}
\forall J:\frac{\sum_{k=1}^{J+1}\left\vert E_{k}\right\vert }{m}\geq\frac
{m-J}{m}=1-\frac{J}{m}=1-\frac{1}{K\left(  N_{1}+N_{2}-2\right)  }J.
\label{eq04012}%
\end{equation}
Consider now $\frac{\sum_{k=1}^{J+1}\left(  \deg\left(  U_{k}\right)  \right)
^{2}}{\left(  2m\right)  ^{2}}$. Each $U_{k}$ has no more than $\frac{n}%
{J}=\frac{K\left(  N_{1}+N_{2}\right)  }{J}$ nodes and each node has degree at
most 2. Hence%
\begin{equation}
\forall J:\frac{\sum_{k=1}^{J+1}\left(  \deg\left(  U_{k}\right)  \right)
^{2}}{\left(  2m\right)  ^{2}}\leq\frac{\left(  J+1\right)  \cdot\left(
2\frac{K\left(  N_{1}+N_{2}\right)  }{J}\right)  ^{2}}{4K^{2}\left(
N_{1}+N_{2}-2\right)  ^{2}}\leq\frac{2\left(  N_{1}+N_{2}\right)  ^{2}%
}{\left(  N_{1}+N_{2}-2\right)  ^{2}}J^{-1} \label{eq04013}%
\end{equation}
(since $\forall J\in\mathbb{N}:\frac{J+1}{J}\leq2$). Combining (\ref{eq04012}%
)\ and (\ref{eq04013})\ we get (\ref{eq04011}).
\end{proof}

Hence, to ensure $Q_{N}\left(  \mathbf{U}_{K,N_{1}N_{2},J},G_{K,N_{1}N_{2}%
}\right)  >Q_{N}\left(  \mathbf{V}_{K,N_{1}N_{2}},G_{K,N_{1}N_{2}}\right)  $
(i.e., that the natural clustering $\mathbf{V}_{K,N_{1}N_{2}}$ has lower
modularity than $\mathbf{U}_{K,N_{1}N_{2},J}$) it suffices to select
$K,N_{1},N_{2},J$ appropriately and use Lemmas \ref{prp0401} and
\ref{prp0402}. A sufficient condition, obtained from (\ref{eq04002}) and
(\ref{eq04011}), is
\begin{equation}
1-\frac{1}{K\left(  N_{1}+N_{2}-2\right)  }J-\frac{2\left(  N_{1}%
+N_{2}\right)  ^{2}}{\left(  N_{1}+N_{2}-2\right)  ^{2}}J^{-1}>1-\frac{\left(
N_{1}-1\right)  ^{2}+\left(  N_{2}-1\right)  ^{2}}{K\left(  N_{1}%
+N_{2}-2\right)  ^{2}}. \label{eq04021}%
\end{equation}
Inspecting (\ref{eq04021}) we see that one way to satisfy it is by fixing
$N_{1}$ and letting $J$ be \textquotedblleft sufficiently
larger\textquotedblright\ than $K$ and $N_{2}$ \textquotedblleft sufficiently
larger\textquotedblright\ than $J$. This is the main idea used in the proof of
the following theorem.

\begin{theorem}
\label{prp0403}For every $K\in\mathbb{N}$ and $\varepsilon\in\left(
0,\frac{1}{2K}\right)  $ there exist $N_{1},N_{2},J\in\mathbb{N}$ (depending
on $\varepsilon$ and $K$) such that%
\begin{align}
Q_{N}\left(  \mathbf{V}_{K,N_{1},N_{2}},G_{K,N_{1},N_{2}}\right)   &
<1-\frac{1}{2K}<1-\varepsilon<Q_{N}\left(  \mathbf{U}_{K,N_{1},N_{2}%
,J},G_{K,N_{1},N_{2}}\right)  ,\label{eq04022}\\
S\left(  \mathbf{V}_{K,N_{1},N_{2}},\mathbf{U}_{K,N_{1},N_{2},J}\right)   &
<\varepsilon. \label{eq04023}%
\end{align}

\end{theorem}

\begin{proof}
Take any $K$ and let $N_{1}=3$,$\ J=xK$,$\ N_{2}=x^{2}K$ (with $x\in
\mathbb{N}$). To prove (\ref{eq04022}) note that%
\begin{align*}
Q_{N}\left(  \mathbf{V}_{K,N_{1},N_{2}},G_{K,N_{1},N_{2}}\right)   &
=1-\frac{4+\left(  x^{2}K-1\right)  ^{2}}{K\left(  1+x^{2}K\right)  ^{2}},\\
Q_{N}\left(  \mathbf{U}_{K,N_{1},N_{2},J},G_{K,N_{1},N_{2}}\right)   &
\geq1-\frac{x}{\left(  1+x^{2}K\right)  }-\frac{2\left(  3+x^{2}K\right)
^{2}}{\left(  1+x^{2}K\right)  ^{2}xK}.
\end{align*}
Define $z=\frac{1}{x}$; then we have $x=\frac{1}{z}$ and
\begin{equation}
Q_{N}\left(  \mathbf{V}_{K,N_{1},N_{2}},G_{K,N_{1},N_{2}}\right)
=1-\frac{4+\left(  x^{2}K-1\right)  ^{2}}{K\left(  1+x^{2}K\right)  ^{2}%
}=1-\frac{4+\left(  \left(  1/z\right)  ^{2}K-1\right)  ^{2}}{K\left(
1+\left(  1/z\right)  ^{2}K\right)  ^{2}}. \label{eq04025a}%
\end{equation}
We can simplify the final $Q_{N}\left(  \mathbf{V}_{K,N_{1},N_{2}}%
,G_{K,N_{1},N_{2}}\right)  $ expression of (\ref{eq04025a}) and write it as
the following function%
\[
f_{1}\left(  z\right)  =\frac{K^{3}-K^{2}+2\left(  K+K^{2}\right)
z^{2}+\left(  K-5\right)  z^{4}}{K\left(  z^{2}+K\right)  ^{2}}.
\]
Now, $1-\frac{4+\left(  \left(  1/z\right)  ^{2}K-1\right)  ^{2}}{K\left(
1+\left(  1/z\right)  ^{2}K\right)  ^{2}}$ has a removable singularity at
$z_{0}=0$, but for every other $z\in\mathbb{R}$ it is identical to
$f_{1}\left(  z\right)  $. We can expand $f_{1}\left(  z\right)  $ in a Taylor
series around $z_{0}=0$ which will also hold for $Q_{N}\left(  \mathbf{V}%
_{K,N_{1},N_{2}},G_{K,N_{1},N_{2}}\right)  $. Hence around $z_{0}=0$ we have
\[
Q_{N}\left(  \mathbf{V}_{K,N_{1},N_{2}},G_{K,N_{1},N_{2}}\right)  =1-\frac
{1}{K}+r_{1}\left(  z\right)  ,
\]
where $r_{1}\left(  z\right)  =a_{2}z^{2}+a_{3}z^{3}+...$ and, from the Taylor
series remainder theorem,\ there exists a constant $A$ such that, for $z$
close to zero, we have
\[
\left\vert r_{1}\left(  z\right)  \right\vert <Az^{2}.
\]
Then, for large finite $x$ (and, in particular, for $x>\sqrt{2KA}$) we have
\begin{equation}
Q_{N}\left(  \mathbf{V}_{K,N_{1},N_{2}},G_{K,N_{1},N_{2}}\right)
=1-\frac{4+\left(  x^{2}K-1\right)  ^{2}}{K\left(  1+x^{2}K\right)  ^{2}%
}<1-\frac{1}{K}+\frac{A}{x^{2}}<1-\frac{1}{2K}. \label{eq04025}%
\end{equation}
Similarly (with $z=\frac{1}{x}$)\ we have \
\begin{align}
Q_{N}\left(  \mathbf{U}_{K,N_{1},N_{2},J},G_{K,N_{1},N_{2}}\right)   &
=1-\frac{xK}{K\left(  1+x^{2}K\right)  }-\frac{2\left(  3+x^{2}K\right)  ^{2}%
}{\left(  1+x^{2}K\right)  ^{2}xK}\nonumber\\
&  =1-\frac{\left(  1/z\right)  }{\left(  1+\left(  1/z\right)  ^{2}K\right)
}-\frac{2\left(  3+\left(  1/z\right)  ^{2}K\right)  ^{2}}{\left(  1+\left(
1/z\right)  ^{2}K\right)  ^{2}\left(  1/z\right)  K} \label{eq04025b}%
\end{align}
Again, we can rewrite the final $Q_{N}\left(  \mathbf{U}_{K,N_{1},N_{2}%
,J},G_{K,N_{1},N_{2}}\right)  $ expression of (\ref{eq04025b}) as
\[
f_{2}\left(  z\right)  =\frac{K^{3}-3K^{2}z+2K^{2}z^{2}-13Kz^{3}%
+Kz^{4}-18z^{5}}{K\left(  z^{2}+K\right)  ^{2}}%
\]
and $1-\frac{\left(  1/z\right)  }{\left(  1+\left(  1/z\right)  ^{2}K\right)
}-\frac{2\left(  3+\left(  1/z\right)  ^{2}K\right)  ^{2}}{\left(  1+\left(
1/z\right)  ^{2}K\right)  ^{2}\left(  1/z\right)  K}$ has a removable
singularity at $z_{0}=0$, but for every other $z\in\mathbb{R}$ it is identical
to $f_{2}\left(  z\right)  $. Hence we can expand $f_{2}\left(  z\right)  $ in
a Taylor series around $z_{0}=0$ which will also hold for $Q_{N}\left(
\mathbf{U}_{K,N_{1},N_{2},J},G_{K,N_{1},N_{2}}\right)  $. Hence around
$z_{0}=0$ we have%
\[
Q_{N}\left(  \mathbf{U}_{K,N_{1},N_{2},J},G_{K,N_{1},N_{2}}\right)
=1-\frac{3}{K}z+r_{2}\left(  z\right)
\]
where $r_{2}\left(  z\right)  =b_{3}z^{3}+b_{4}z^{4}+...$ and there exists a
constant $B$ such that, for $z$ close to zero, we have
\[
\left\vert r_{2}\left(  z\right)  \right\vert <Bz^{3}<Bz^{2};
\]
this in turn implies that
\[
r_{2}\left(  z\right)  >-Bz^{2}.
\]
Then, for large $x$ (and, in particular, for $x>KB$) we have \
\begin{align}
Q_{N}\left(  \mathbf{U}_{K,N_{1},N_{2},J},G_{K,N_{1},N_{2}}\right)   &
=1-\frac{xK}{K\left(  1+x^{2}K\right)  }-\frac{2\left(  3+x^{2}K\right)  ^{2}%
}{\left(  1+x^{2}K\right)  ^{2}xK}\nonumber\\
&  >1-\frac{3}{Kx}-\frac{B}{x^{2}}>1-\frac{4}{Kx}. \label{eq04026}%
\end{align}
For any $\varepsilon\in\left(  0,\frac{1}{2K}\right)  $, choose any
$x>\max\left(  \frac{4}{K\varepsilon},\sqrt{2KA},KB\right)  $; then we have
$\frac{1}{2K}>\varepsilon>\frac{4}{Kx}$ which, combined with (\ref{eq04025})
and (\ref{eq04026}), gives
\[
Q_{N}\left(  \mathbf{U}_{K,N_{1},N_{2},J},G_{K,N_{1},N_{2}}\right)
>1-\frac{4}{Kx}>1-\varepsilon>1-\frac{1}{2K}>Q_{N}\left(  \mathbf{V}%
_{K,N_{1},N_{2}},G_{K,N_{1},N_{2}}\right)  .
\]
In short, we can satisfy (\ref{eq04022}) for every $K\in\mathbb{N}$ and every
$\varepsilon\in\left(  0,\frac{1}{2K}\right)  $,\ by taking $x$
\textquotedblleft sufficiently large\textquotedblright\ and $N_{1}=3$%
,$\ J=xK$,$\ N_{2}=x^{2}K$.

We now turn to (\ref{eq04023}). Let $b$ (resp. $c$)\ be the number of node
pairs in the same cluster under $\mathbf{U}_{K,N_{1},N_{2},J}$ (resp. under
$\mathbf{V}_{K,N_{1},N_{2}}$). We obviously have $b=a_{01}+a_{11}\geq a_{11}$
and $a_{10}+a_{01}+a_{11}\geq a_{10}+a_{11}=c>0$. Hence%
\[
S\left(  \mathbf{U}_{K,N_{1},N_{2},J},\mathbf{V}_{K,N_{1},N_{2}}\right)
=\frac{a_{11}}{a_{10}+a_{01}+a_{11}}\leq\frac{b}{c}.
\]
We first obtain an upper bound for $b$. Since each $U_{j}$ contains no more
than $L=\frac{n}{J}$ nodes , the number of node pairs that can be formed in
$U_{j}$ is no more than $\frac{\left(  \frac{n}{J}\right)  \left(  \frac{n}%
{J}-1\right)  }{2}<\frac{n^{2}/2}{J^{2}}$. Also, $n=K\left(  N_{1}%
+N_{2}\right)  $ so, for big $N_{2}$, $\frac{n^{2}/2}{J^{2}}<\frac{\left(
2KN_{2}\right)  ^{2}}{J^{2}}$. There are at most $J+1$ clusters, so we have
\[
b<\left(  J+1\right)  \frac{\left(  2KN_{2}\right)  ^{2}}{J^{2}}=\left(
xK+1\right)  \frac{\left(  2Kx^{2}K\right)  ^{2}}{\left(  xK\right)  ^{2}%
}=4K^{3}x^{3}+4K^{2}x^{2}.
\]
Next we compute $c$. In $\mathbf{V}_{K,N_{1},N_{2}}$ there exist $K$ clusters
of $N_{1}=3$ nodes and each cluster has $\frac{N_{1}\left(  N_{1}-1\right)
}{2}=3$ node pairs; there also exist $K$ clusters of $N_{2}\ $nodes and each
cluster has $\frac{N_{2}\left(  N_{2}-1\right)  }{2}$ node pairs. We have
\[
c=3K+K\frac{N_{2}\left(  N_{2}-1\right)  }{2}=3K+K\frac{x^{2}K\left(
x^{2}K-1\right)  }{2}=\allowbreak\frac{1}{2}K^{3}x^{4}-\frac{1}{2}K^{2}%
x^{2}+3K\allowbreak.
\]

And so we have
\begin{align*}
0  &  \leq S\left(  \mathbf{U}_{K,N_{1},N_{2},J},\mathbf{V}_{K,N_{1},N_{2}%
}\right)  <\frac{4K^{3}x^{3}+4K^{2}x^{2}}{\frac{1}{2}K^{3}x^{4}-\frac{1}%
{2}K^{2}x^{2}+3K}\Rightarrow\\
0  &  \leq\lim_{x\rightarrow\infty}S\left(  \mathbf{U}_{K,N_{1},N_{2}%
,J},\mathbf{V}_{K,N_{1},N_{2}}\right)  \leq\lim_{x\rightarrow\infty}%
\frac{4K^{3}x^{3}+4K^{2}x^{2}}{\frac{1}{2}K^{3}x^{4}-\frac{1}{2}K^{2}x^{2}%
+3K}=0.
\end{align*}
Hence, for every $\varepsilon>0$ and $x$ sufficiently large, (\ref{eq04023}%
)\ is satisfied.
\end{proof}

We see from (\ref{eq04022}) that we can always find a clustering
$\mathbf{U}_{K,N_{1},N_{2},J}$ which achieves higher modularity than the
natural clustering $\mathbf{V}_{K,N_{1},N_{2}}$ and, in fact, greater than
$1-\varepsilon$, where $\varepsilon$ can get arbitrarily small independently
of $K$. On the other hand, $Q_{N}\left(  \mathbf{V}_{K,N_{1},N_{2}}%
,G_{K,N_{1},N_{2}}\right)  $ is no greater than $1-\frac{1}{2K}$; for small
$K$ values this can be appreciably less than one. In other words, we can
choose $K$ so that $G_{K,N_{1},N_{2}}$ \ does not have very high
\textquotedblleft natural modularity\textquotedblright\ but its
\textquotedblleft artificial modularity\textquotedblright\ (the one achieved
by the pair $\left(  \mathbf{U}_{K,N_{1},N_{2},J},G_{K,N_{1},N_{2}}\right)  $
) can be arbitrarily close to one.

We see from (\ref{eq04023}) that, with respect to the Jaccard similarity
criterion, $\mathbf{U}_{K,N_{1},N_{2},J}$ is very different from
$\mathbf{V}_{K,N_{1},N_{2}}$. We could have reached a similar conclusion in a
simpler manner. Recall that the number of clusters of $\mathbf{U}%
_{K,N_{1},N_{2},J}$ is at least $J=xK$ and we can choose $x$ arbitrarily
large; on the other hand, $\mathbf{V}_{K,N_{1},N_{2}}$ has $2K$ clusters.
Intuitively, $\mathbf{U}_{K,N_{1},N_{2},J}$ must be very different from
$\mathbf{V}_{K,N_{1},N_{2}}$, since the ratio of their cluster number is
$\frac{x}{2}$ and $x$ can become arbitrarily large (of course the Jaccard
similarity index captures this fact in a more precise manner).

Let $\mathbf{V}^{\ast}=\arg\max_{V\in\mathcal{V}}Q_{N}\left(  \mathbf{V,}%
G_{K,N_{1},N_{2}}\right)  $. While it is conceivable that $\mathbf{V}^{\ast}$
is more similar (in the Jaccard sense)\ to $\mathbf{V}_{K,N_{1},N_{2}}$ than
to some $\mathbf{U}_{K,N_{1},N_{2},J}$, this seems unlikely. In light of the
remarks of Section \ref{sec0303}, it is more likely that $\mathbf{V}^{\ast}$
will have many more clusters than $\mathbf{V}$. In other words, it appears
that, for the graphs $G_{K,N_{1},N_{2}}$, modularity maximization leads to an
overestimation of the number of clusters, i.e., we have a case of modularity
\textquotedblleft\emph{over-resolution}\textquotedblright.

The bounds utilized in Lemmas \ref{prp0401}-\ref{prp0402} and Theorem
\ref{prp0403} are quite conservative. In many cases the inequality
\begin{equation}
Q_{N}\left(  \mathbf{V}_{K,N_{1},N_{2}},G_{K,N_{1},N_{2}}\right)
<Q_{N}\left(  \mathbf{U}_{K,N_{1},N_{2},J},G_{K,N_{1},N_{2}}\right)
\label{eq04028}%
\end{equation}
is attained even when the abovementioned bounds are not satisfied. This can be
seen in Table 1, which has been compiled by taking fixed $K=3$, $N_{1}=3$ and
using several $x$ values (recall that $J=xK$, $N_{2}=x^{2}K$). The first six
entries of each column list the quantities used in the proof of Theorem
\ref{prp0403} and, for \textquotedblleft sufficiently large\textquotedblright%
\ $x$, should form an increasing sequence, in accordance to the inequalities
(\ref{eq04021})-(\ref{eq04022})\ and (\ref{eq04025})-(\ref{eq04026}). This is
indeed the case for $x=8$ and $x=10$; on the other hand, for $x=4$ and $x=6$
one inequality is violated (between the third and fourth row)\ but
(\ref{eq04028}) still holds.

\begin{center}%
\begin{tabular}
[c]{|l|l|l|l|l|l|}\hline
Row no. &  & $x=4$ & $x=6$ & $x=8$ & $x=10$\\\hline
1 & $Q_{N}\left(  \mathbf{V}_{K,N_{1},N_{2}},G_{K,N_{1},N_{2}}\right)  $ &
0.6928 & 0.6787 & 0.6735 & 0.6711\\\hline
2 & $1-\frac{\left(  N_{1}-1\right)  ^{2}+\left(  N_{2}-1\right)  ^{2}%
}{K\left(  N_{1}+N_{2}-2\right)  ^{2}}$ & 0.6928 & 0.6787 & 0.6735 &
0.6711\\\hline
3 & $1-\frac{1}{2K}$ & 0.8333 & 0.8333 & 0.8333 & 0.8333\\\hline
4 & $1-\frac{4}{Kx}$ & 0.6667 & 0.7778 & 0.8333 & 0.8667\\\hline
5 & $1-\frac{1}{K\left(  N_{1}+N_{2}-2\right)  }J-\frac{2\left(  N_{1}%
+N_{2}\right)  ^{2}}{\left(  N_{1}+N_{2}-2\right)  ^{2}}J^{-1}$ & 0.7378 &
0.8297 & 0.8735 & 0.8992\\\hline
6 & $Q_{N}\left(  \mathbf{U}_{K,N_{1},N_{2},J},G_{K,N_{1},N_{2}}\right)  $ &
0.8407 & 0.8915 & 0.9179 & 0.9340\\\hline
7 & $S\left(  \mathbf{U}_{K,N_{1},N_{2},J},\mathbf{V}_{K,N_{1},N_{2}}\right)
$ & 0.2196 & 0.1540 & 0.1190 & 0.0967\\\hline
\end{tabular}

\textbf{Table 1. }Several quantities appearing in the proof of Theorem
\textbf{ }\ref{prp0403}. In each column and for rows 1 to 6, for large enough
$x$, the value of each row must be no less than that of the previous one.
\end{center}

From rows 1 and 6 of Table 1 we see that $Q_{N}\left(  \mathbf{V}%
_{K,N_{1},N_{2}},G_{K,N_{1},N_{2}}\right)  $ is a decreasing and $Q_{N}\left(
\mathbf{U}_{K,N_{1},N_{2},J},G_{K,N_{1},N_{2}}\right)  $ an increasing
function of $x$ . From row 7 we see that the the Jaccard similarity is a
decreasing function of $x$. These observations verify straightforward
conclusions which can be drawn from the proof of Theorem \ref{prp0403}.

\subsection{Second Example\label{sec0402}}

It might be argued that the results of Section \ref{sec0401} are only possible
because we have used the disconnected graphs $G_{K,N_{1},N_{2}}$. This is not
the case. In this section we will illustrate the same issues using the family
of \emph{connected }graphs $H_{K,N_{1},N_{2}}$ illustrated in Figure
\ref{fig0402}. We start with connected $H_{N_{1},N_{2}}$ graphs, each of which
is a path of $N_{1}+N_{2}$ nodes, with extra edges added between the first
$N_{1}$ (resp. the second $N_{2}$) nodes at distance two of each other. Then
we construct the $H_{K,N_{1},N_{2}}$ graphs by joining in series $K$
$H_{N_{1},N_{2}}$ subgraphs. \begin{figure}[h]
\centering\scalebox{0.6}{\includegraphics{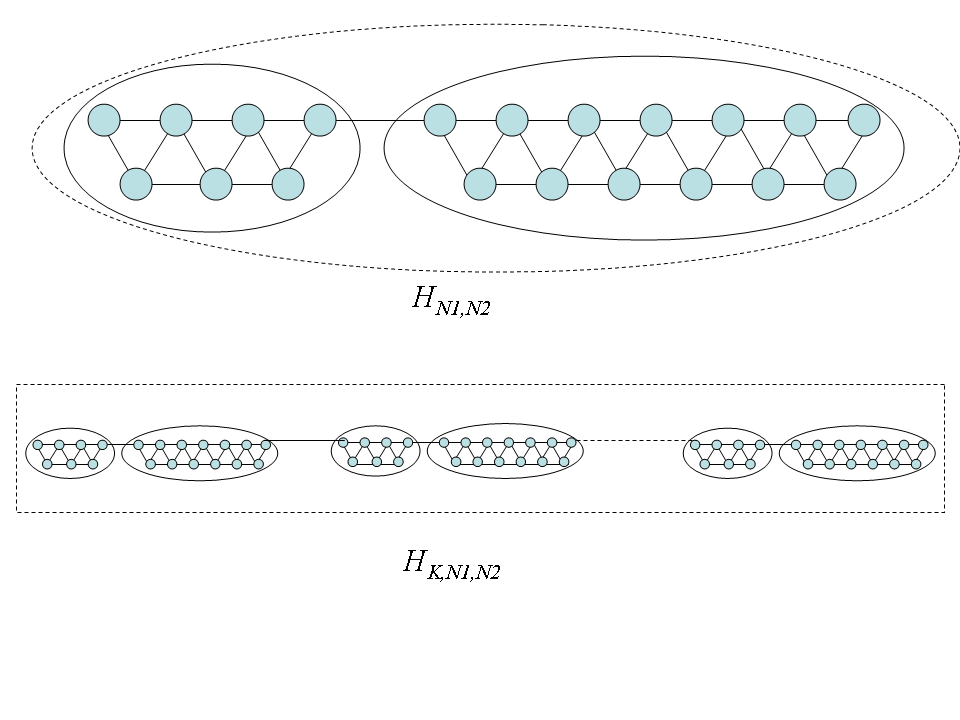}}\caption[Graph Family
A]{Graph Family $H_{K,N_{1},N_{2}}.$}%
\label{fig0402}%
\end{figure}

We will use the same clusterings $\mathbf{V}_{K,N_{1},N_{2}}\ $and clustering
sequences $\left\{  \mathbf{U}_{K,N_{1},N_{2}J}\right\}  _{J=1}^{n}$ as in
Section \ref{sec0401}. Once again, for reasons similar to the ones discussed
in Section \ref{sec0401}, we claim that $\mathbf{V}_{K,N_{1},N_{2}}$ is the
natural clustering of $H_{K,N_{1},N_{2}}$. Namely, cluster boundaries should
occur across edges incident on the most weakly connected nodes; this shows
that the $V_{K,N_{1},N_{2},k}$ clusters must be preserved; any partition of
$V_{K,N_{1},N_{2},k}$ into finer clusters cannot be justified, since all of
its edges have the same connectivity pattern. Hence $\mathbf{V}_{K,N_{1}%
,N_{2}}$ is the \textquotedblleft intuitively best\textquotedblright\ (i.e.,
the \textquotedblleft natural\textquotedblright) clustering of $H_{K,N_{1}%
,N_{2}}$.

Once again, we obtain (in three steps) a result similar to Theorem
\ref{prp0403}. First we need two lemmas.

\begin{lemma}
\label{prp0404}For every $K,N_{1},N_{2}\in\mathbb{N}$ with $N_{1},N_{2}%
\geq5\mathbb{\ }$we have%
\begin{equation}
Q_{N}\left(  \mathbf{V}_{K,N_{1},N_{2}},H_{K,N_{1},N_{2}}\right)
<1-\frac{K\left(  \left(  4N_{1}-8\right)  ^{2}+\left(  4N_{2}-8\right)
^{2}\right)  }{\left(  4K\left(  N_{1}+N_{2}-2\right)  \right)  ^{2}}.
\label{eq04051}%
\end{equation}

\end{lemma}

\begin{proof}
We fix $K,N_{1},N_{2}$ and, for brevity, we write $H$ for $H_{K,N_{1},N_{2}}$
and $\mathbf{V}$ for $\mathbf{V}_{K,N_{1},N_{2}}$; $\mathbf{V}^{\prime}$ and
$\mathbf{V}^{\prime\prime}$ have the same meaning as previously. In
$\mathbf{V}_{K,N_{1},N_{2}}$ there exist $2K-1$ extracluster edges, so we
have
\begin{equation}
\frac{\sum_{k=1}^{2K}\left\vert E_{k}\right\vert }{m}<1. \label{eq04054}%
\end{equation}
For each $V_{k}\in\mathbf{V}^{\prime}$, \ there are two border nodes on the
left, two border nodes on the right and $N_{1}-4$ inner nodes. Each of the
inner nodes has degree 4; each of the border nodes has degree 3, except for
the first and last node of the graph, which have degree 2. Hence for each
$V_{k}\in\mathbf{V}^{\prime}$ we have the bounds \
\[
\left(  N_{1}-4\right)  \cdot4+4\cdot2=4N_{1}-8<\deg\left(  V_{k}\right)
<4N_{1}-4=\text{ }\left(  N_{1}-4\right)  \cdot4+4\cdot3.
\]
Similarly, for each $V_{k}\in\mathbf{V}^{\prime\prime}$ we have the bounds%
\[
\left(  N_{2}-4\right)  \cdot4+4\cdot2=4N_{2}-8<\deg\left(  V_{k}\right)
<\text{ }4N_{2}-4=\left(  N_{2}-4\right)  \cdot4+4\cdot3.
\]
The total number of edges is $m=\frac{\sum_{k=1}^{2K}\deg\left(  V_{k}\right)
}{2}$ and we have
\begin{align}
\frac{K\left(  4N_{1}-8+4N_{2}-8\right)  }{2}  &  <\frac{\sum_{k=1}^{2K}%
\deg\left(  V_{k}\right)  }{2}<\frac{K\left(  4N_{1}-4+4N_{2}-4\right)  }%
{2}\Rightarrow\nonumber\\
2K\left(  N_{1}+N_{2}-4\right)   &  <m<2K\left(  N_{1}+N_{2}-2\right)  .
\label{eq04053}%
\end{align}
In addition we have%
\begin{equation}
K\left(  \left(  4N_{1}-8\right)  ^{2}+\left(  4N_{2}-8\right)  ^{2}\right)
<\sum_{k=1}^{2K}\left(  \deg\left(  V_{k}\right)  \right)  ^{2}<K\left(
\left(  4N_{1}-4\right)  ^{2}+\left(  4N_{2}-4\right)  ^{2}\right)  .
\label{eq04052}%
\end{equation}
Combining (\ref{eq04053}) and (\ref{eq04052}) we get
\begin{equation}
\frac{\sum_{k=1}^{2K}\left(  \deg\left(  V_{k}\right)  \right)  ^{2}}{\left(
2m\right)  ^{2}}>\frac{K\left(  \left(  4N_{1}-8\right)  ^{2}+\left(
4N_{2}-8\right)  ^{2}\right)  }{\left(  4K\left(  N_{1}+N_{2}-2\right)
\right)  ^{2}}. \label{eq04055}%
\end{equation}
Finally, combining (\ref{eq04054}) and (\ref{eq04055})\ we get the required bound.
\end{proof}

\begin{lemma}
\label{prp0405}For every $K,N_{1},N_{2},J\in\mathbb{N}$ with $N_{1},N_{2}%
\geq5\mathbb{\ }$and $J\leq n=K\left(  N_{1}+N_{2}\right)  $ we have%
\begin{equation}
Q_{N}\left(  \mathbf{U}_{K,N_{1},N_{2},J},H_{K,N_{1},N_{2}}\right)
>1-\frac{3}{2K\left(  N_{1}+N_{2}-4\right)  }J-\frac{2\left(  N_{1}%
+N_{2}\right)  ^{2}}{\left(  N_{1}+N_{2}-4\right)  ^{2}}J^{-1}.
\label{eq04061}%
\end{equation}

\end{lemma}

\begin{proof}
Extracluster edges in $\mathbf{U}_{J}$ can only occur between successive
clusters\footnote{There is an exception when $J=n$, but in this case too
(\ref{eq04062}) holds.} $U_{k}$, $U_{k+1}$; between any such pair there exist
at most three such edges; hence $\mathbf{U}_{J}$ cannot have more than $3J$
extracluster edges. Consequently%
\begin{equation}
\frac{\sum_{k=1}^{J+1}\left\vert E_{k}\right\vert }{m}\geq\frac{m-3J}%
{m}=1-\frac{3J}{m}>1-\frac{3J}{2K\left(  N_{1}+N_{2}-4\right)  }.
\label{eq04062}%
\end{equation}
Each $U_{k}$ has at most $\frac{n}{J}=\frac{K\left(  N_{1}+N_{2}\right)  }{J}$
nodes and each node has degree at most 4. Hence
\begin{equation}
\frac{\sum_{k=1}^{J+1}\left(  \deg\left(  U_{k}\right)  \right)  ^{2}}{\left(
2m\right)  ^{2}}\leq\frac{\left(  J+1\right)  \left(  4\frac{K\left(
N_{1}+N_{2}\right)  }{J}\right)  ^{2}}{\left(  4K\left(  N_{1}+N_{2}-4\right)
\right)  ^{2}}\leq\frac{2\left(  N_{1}+N_{2}\right)  ^{2}}{\left(  N_{1}%
+N_{2}-4\right)  ^{2}}J^{-1}. \label{eq04063}%
\end{equation}
Combining (\ref{eq04062})\ and (\ref{eq04063})\ we get the bound
(\ref{eq04061}).
\end{proof}

To ensure $Q_{N}\left(  \mathbf{U}_{K,N_{1}N_{2},J},H_{K,N_{1}N_{2}}\right)
>Q_{N}\left(  \mathbf{V}_{K,N_{1}N_{2}},H_{K,N_{1}N_{2}}\right)  $ it suffices
to choose appropriate $K,N_{1},N_{2},J$ and use Lemmas \ref{prp0404} and
\ref{prp0405}. A sufficient condition, obtained from (\ref{eq04051}) and
(\ref{eq04061}), is
\begin{equation}
1-\frac{3}{2K\left(  N_{1}+N_{2}-4\right)  }J-\frac{2\left(  N_{1}%
+N_{2}\right)  ^{2}}{\left(  N_{1}+N_{2}-4\right)  ^{2}}J^{-1}>1-\frac
{K\cdot\left(  \left(  4N_{1}-8\right)  ^{2}+\left(  4N_{2}-8\right)
^{2}\right)  }{\left(  4K\left(  N_{1}+N_{2}-2\right)  \right)  ^{2}}.
\label{eq04071}%
\end{equation}
Now we can prove the following.

\begin{theorem}
\label{prp0406}For every $K\in\mathbb{N}$ and $\varepsilon\in\left(
0,\frac{1}{2K}\right)  $ there exist $N_{1},N_{2},J\in\mathbb{N}$ (depending
on $\varepsilon,K$)\ such that
\begin{align}
Q_{N}\left(  \mathbf{V}_{K,N_{1},N_{2}},H_{K,N_{1},N_{2}}\right)   &
<1-\frac{1}{2K}<1-\varepsilon<Q_{N}\left(  \mathbf{U}_{K,N_{1},N_{2}%
,J},H_{K,N_{1},N_{2}}\right) \label{eq04072}\\
S\left(  \mathbf{V}_{K,N_{1},N_{2}},\mathbf{U}_{K,N_{1},N_{2},J}\right)   &
<\varepsilon. \label{eq04073}%
\end{align}

\end{theorem}

\begin{proof}
Take any $K$. Letting $N_{1}=6$,$\ J=xK$,$\ N_{2}=x^{2}K$ we have
\begin{align*}
Q_{N}\left(  \mathbf{V}_{K,N_{1},N_{2}},H_{K,N_{1},N_{2}}\right)   &
<1-\frac{K\left(  16^{2}+\left(  4x^{2}K-8\right)  ^{2}\right)  }{\left(
4K\left(  4+x^{2}K\right)  \right)  ^{2}},\\
Q_{N}\left(  \mathbf{U}_{K,N_{1},N_{2},J},H_{K,N_{1},N_{2}}\right)   &
>1-\frac{3x}{2\left(  2+x^{2}K\right)  }-\frac{2\left(  6+x^{2}K\right)  ^{2}%
}{xK\left(  2+x^{2}K\right)  ^{2}}.
\end{align*}
Defining $z=\frac{1}{x}$ we have $x=\frac{1}{z}$ and
\begin{equation}
1-\frac{K\left(  16^{2}+\left(  4x^{2}K-8\right)  ^{2}\right)  }{\left(
4K\left(  4+x^{2}K\right)  \right)  ^{2}}=1-\frac{K\left(  16^{2}+\left(
4\left(  1/z\right)  ^{2}K-8\right)  ^{2}\right)  }{\left(  4K\left(
4+\left(  1/z\right)  ^{2}K\right)  \right)  ^{2}}. \label{eq04073a}%
\end{equation}
Similarly to the proof of Theorem \ref{prp0403}, there is a function
$f_{3}\left(  z\right)  $ which, for every $z\neq z_{0}=0$, \ is equal to the
right part of (\ref{eq04073a})\ and around $z_{0}$ has the Taylor expansion
\[
f_{3}\left(  z\right)  =1-\frac{1}{K}+r_{3}\left(  z\right)
\]
where $r_{3}\left(  z\right)  =c_{2}z^{2}+c_{3}z^{2}+...$ . Furthermore, there
exists a constant $C$ such that, for $z$ close to zero, we have%
\[
\left\vert r_{3}\left(  z\right)  \right\vert <Cz^{2}.
\]
Then, for large $x$ (and in particular for $x>\sqrt{2KC}$) we have
\begin{equation}
Q_{N}\left(  \mathbf{V}_{K,N_{1},N_{2}},H_{K,N_{1},N_{2}}\right)  <1-\frac
{1}{K}+\frac{C}{x^{2}}<1-\frac{1}{2K}. \label{eq04074}%
\end{equation}
Similarly, with $z=1/x$, we have
\begin{align}
Q_{N}\left(  \mathbf{U}_{K,N_{1},N_{2},J},H_{K,N_{1},N_{2}}\right)   &
>1-\frac{3xK}{2K\left(  2+x^{2}K\right)  }-\frac{2\left(  6+x^{2}K\right)
^{2}}{xK\left(  2+x^{2}K\right)  ^{2}}\nonumber\\
&  =1-\frac{3\left(  1/z\right)  K}{2K\left(  2+\left(  1/z\right)
^{2}K\right)  }-\frac{2\left(  6+\left(  1/z\right)  ^{2}K\right)  ^{2}%
}{\left(  1/z\right)  K\left(  2+\left(  1/z\right)  ^{2}K\right)  ^{2}}%
=f_{4}\left(  z\right)  . \label{eq04073b}%
\end{align}
Once again, there is a function $f_{4}\left(  z\right)  $ which, for every
$z\neq z_{0}=0$, \ is equal to the right part of (\ref{eq04073b})\ and around
$z_{0}$ has the Taylor expansion
\[
f_{4}\left(  z\right)  =1-\frac{7}{2K}z+r_{4}\left(  z\right)
\]
where $r_{4}\left(  z\right)  =d_{3}z^{3}+d_{4}z^{4}+...$ . And there exists a
constant $D$ such that, for $z$ close to zero, we have%
\[
\left\vert r_{4}\left(  z\right)  \right\vert <Dz^{3}<Dz^{2},\quad
r_{4}\left(  z\right)  >-Dz^{2}.
\]
Then, for large $x$ (an, in particular, for $x>2KD$) we have%
\begin{align}
Q_{N}\left(  \mathbf{U}_{K,N_{1},N_{2},J},H_{K,N_{1},N_{2}}\right)   &
>1-\frac{3xK}{2K\left(  2+x^{2}K\right)  }-\frac{2\left(  6+x^{2}K\right)
^{2}}{xK\left(  2+x^{2}K\right)  ^{2}}\nonumber\\
&  >1-\frac{7}{2Kx}-\frac{D}{x^{2}}>1-\frac{4}{Kx}. \label{eq04075}%
\end{align}
For any $\varepsilon\in\left(  0,\frac{1}{2K}\right)  $ choose any
$x>\max\left(  \frac{4}{K\varepsilon},\sqrt{2KC},2KD\right)  $; then we have
$\frac{1}{2K}>\varepsilon>\frac{4}{Kx}$ which, combined with (\ref{eq04074})
and (\ref{eq04075}), yields%
\[
Q_{N}\left(  \mathbf{U}_{K,N_{1},N_{2},J},G_{K,N_{1},N_{2}}\right)
>1-\frac{4}{Kx}>1-\varepsilon>1-\frac{1}{2K}>Q_{N}\left(  \mathbf{V}%
_{K,N_{1},N_{2}},G_{K,N_{1},N_{2}}\right)  .
\]
In short, we can satisfy (\ref{eq04072}) for every $K\in\mathbb{N}$ and every
$\varepsilon\in\left(  0,\frac{1}{2K}\right)  $,\ by taking $x$
\textquotedblleft sufficiently large\textquotedblright\ and $N_{1}=6$%
,$\ J=xK$,$\ N_{2}=x^{2}K$.

Finally, (\ref{eq04073}) is exactly the same as (\ref{eq04023}) and has
already been proved.
\end{proof}

Similarly to Section \ref{sec0401}, the bounds utilized in Lemmas
\ref{prp0404}-\ref{prp0405} and Theorem \ref{prp0406} are conservative and the
inequality
\begin{equation}
Q_{N}\left(  \mathbf{V}_{K,N_{1},N_{2}},H_{K,N_{1},N_{2}}\right)
<Q_{N}\left(  \mathbf{U}_{K,N_{1},N_{2},J},H_{K,N_{1},N_{2}}\right)
\label{eq04078}%
\end{equation}
can be satisfied even when the bounds are violated. This can be seen in Table
2, which is analogous to Table 1 of Section \ref{sec0401}. We have used $K=3$,
$N_{1}=6$ and several $x$ values. The first six entries of each column list
the quantities used in the proof of Theorem \ref{prp0403} and, for
\textquotedblleft sufficiently large\textquotedblright\ $x$, should form an
increasing sequence. This is the case for $x=8$ and $x=10$; for $x=6$ the
sequence is not increasing but (\ref{eq04078}) holds.

\begin{center}%
\begin{tabular}
[c]{|l|l|l|l|l|}\hline
Row no. &  & $x=6$ & $x=8$ & $x=10$\\\hline
1 & $Q_{N}\left(  \mathbf{V}_{K,N_{1},N_{2}},G_{K,N_{1},N_{2}}\right)  $ &
0.6872 & 0.6787 & 0.6745\\\hline
2 & $1-\frac{K\left(  \left(  4N_{1}-8\right)  ^{2}+\left(  4N_{2}-8\right)
^{2}\right)  }{\left(  4K\left(  N_{1}+N_{2}-2\right)  \right)  ^{2}}$ &
0.7010 & 0.6866 & 0.6796\\\hline
3 & $1-\frac{1}{2K}$ & 0.8333 & 0.8333 & 0.8333\\\hline
4 & $1-\frac{4}{Kx}$ & 0.7778 & 0.8333 & 0.8667\\\hline
5 & $1-\frac{3}{2K\left(  N_{1}+N_{2}-4\right)  }J-\frac{2\left(  N_{1}%
+N_{2}\right)  ^{2}}{\left(  N_{1}+N_{2}-4\right)  ^{2}}J^{-1}$ & 0.7988 &
0.8513 & 0.8819\\\hline
6 & $Q_{N}\left(  \mathbf{U}_{K,N_{1},N_{2},J},G_{K,N_{1},N_{2}}\right)  $ &
0.8743 & 0.8986 & 0.9182\\\hline
7 & $S\left(  \mathbf{U}_{K,N_{1},N_{2},J},\mathbf{V}_{K,N_{1},N_{2}}\right)
$ & 0.1695 & 0.1183 & 0.0956\\\hline
\end{tabular}

\textbf{Table 2. }Several quantities appearing in the proof of Theorem
\textbf{ }\ref{prp0406}. In each column and for rows 1 to 6, for large enough
$x$, the value of each row must be no less than that of the previous one.
\end{center}

From rows 1 and 6 of Table 2 we see that $Q_{N}\left(  \mathbf{V}%
_{K,N_{1},N_{2}},G_{K,N_{1},N_{2}}\right)  $ is decreasing with $x$ and
$Q_{N}\left(  \mathbf{U}_{K,N_{1},N_{2},J},G_{K,N_{1},N_{2}}\right)  $ is
increasing; from row 7 we see that the Jaccard similarity is decreasing with
$x$.

\section{Discussion and Related Work\label{sec05}}

Theorems \ref{prp0403} and \ref{prp0406} cast doubt on the efficacy of
Newman's modularity $Q_{N}$ as \textquotedblleft an objective metric for
choosing the number of communities\textquotedblright\ \cite{newman2004finding}%
. In fact, the existence of such an objective metric can be doubted and the
meaning of the terms \textquotedblleft best clustering\textquotedblright,
\textquotedblleft natural clustering\textquotedblright, etc. are rather
ambiguous, as noted by several researchers (for a discussion see \cite[Section
III]{fortunato2010community}).

Consider for example the term \textquotedblleft good
clustering\textquotedblright. A good clustering $\mathbf{V}$ of a graph $G$
should be objectively recognizable by a high value of $Q\left(  \mathbf{V}%
,G\right)  $, where $Q$ is a \textquotedblleft good\ CQF \textquotedblright.
However, a good CQF is one which assigns high scores to good clusterings.
Hence it appears that the definition of \textquotedblleft good
clusterings\textquotedblright\ and \textquotedblleft good
CQF\textquotedblright\ is a circular process.

While obtaining a \textquotedblleft good CQF\textquotedblright\ is a
worthwhile target, the main focus of the current paper has been the use of
$Q_{N}$ towards estimation of the true number of communities. Since we have
argued that $Q_{N}$'s performance is less than ideal, let us conclude by
discussing alternative ways to perform community number selection.

Let us start by stating that we consider \textquotedblleft community number
selection\textquotedblright\ to be a special case of the general problem of
\textquotedblleft cluster number selection\textquotedblright, which has been
exhaustively studied in the \textquotedblleft classic\textquotedblright%
\ clustering literature (see for example the book \cite{duda1995pattern}). In
this literature, cluster number selection has been recognized as
\textquotedblleft a fundamental, and largely unsolved, problem in cluster
analysis\textquotedblright\cite{sugar2003finding}.

Many works in the the \textquotedblleft classic\textquotedblright\ clustering
literature treat cluster number selection through a two-stage approach. First,
a CQF is used to obtain the optimal clustering of size $K$, for $K\in\left\{
1,2,...,K_{\max}\right\}  $. Then the optimal $K$ value (and hence the overall
optimal clustering)\ is obtained using a \emph{cluster number selection
criterion} \ (CNSC)\footnote{The terms \emph{model order}\ selection criterion
and \emph{cluster validity} selection criterion are also used.\ } such as the
Akaike Information Criterion \cite{akaike1974new}, the Bayesian Information
Criterion \cite{schwarz1978estimating}, Minimum Description Length
\cite{rissanen1978modeling}, the gap statistic \cite{tibshirani2001estimating}%
, the knee criterion \cite{salvador2004determining} etc. Details on CNSC can
be found in
\cite{boutin2004cluster,dimitriadou2002examination,halkidi2002cluster,halkidi2002clustering,milligan1985examination}%
.

The two-stage approach has also been used in the community detection
literature. In Newman's seminal paper \cite{newman2004finding} a dendrogram is
obtained through a divisive hierarchical clustering process (which uses
\emph{betweenness}, rather than $Q_{N}$) and the dendrogram cutoff level (and
hence the number of communities)\ is obtained by maximizaton of $Q_{N}$. In
this case $Q_{N}$ is used as CNSC rather than as a CQF.

We find the two-stage approach to community detection promising and we believe
it deserves further research. In particular, we expect that the two-stage
approach will yield better results if a better CNSC\ than $Q_{N}$ is used.
However the specification of good CNSC's\ must overcome the same difficulties
(associated with circularity) mentioned in the beginning of this section. We
hope these difficulties can be alleviated by the use of an \emph{axiomatic
}approach, which we will report in a future publication.

\bibliographystyle{amsplain}
\providecommand{\bysame}{\leavevmode\hbox to3em{\hrulefill}\thinspace}
\providecommand{\MR}{\relax\ifhmode\unskip\space\fi MR }
% \MRhref is called by the amsart/book/proc definition of \MR.
\providecommand{\MRhref}[2]{%
  \href{http://www.ams.org/mathscinet-getitem?mr=#1}{#2}
}
\providecommand{\href}[2]{#2}

\end{document}